\documentclass[letterpaper, 10 pt, conference]{ieeeconf}  

\IEEEoverridecommandlockouts                              

\usepackage{cite}

\usepackage{amsmath,amssymb,amsfonts}
\usepackage{amsthm}
\usepackage{bbold}
\usepackage{graphicx}
\usepackage{textcomp}
\usepackage{xcolor,subfig,multirow,caption}
\newtheorem{prop}{Proposition}

\newtheorem{theo}{Theorem}
\newtheorem{coro}{Corollary}

\usepackage[norelsize, linesnumbered, ruled, lined, boxed, commentsnumbered, noend]{algorithm2e}

\usepackage[normalem]{ulem}
\usepackage{mathabx}

\ifodd 1
\newcommand{\rev}[1]{{\color{black}#1}} 
\newcommand{\cw}[1]{{\color{black}#1}} 
\newcommand{\com}[1]{{\color{red}\textbf{Comment}:#1}}
\else
\newcommand{\rev}[1]{#1}

\newcommand{\com}[1]{}
\fi

\newcommand{\comment}[1]{}

\title{\LARGE \bf
Decision-Making on Timing and Route Selection:\\ A Game-Theoretic Approach
}
\author{Chenlan Wang$^{1}$
, Mingyan Liu$^{1}$
\thanks{$^{1}$ University of Michigan, Ann Arbor, USA. 
}
}

\begin{document}

\maketitle
\thispagestyle{empty}
\pagestyle{empty}

\begin{abstract}
We present a Stackelberg game model to investigate how individuals make their decisions on timing and route selection. Group formation can naturally result from these decisions, but only when individuals arrive at the same time and choose the same route. Although motivated by bird migration, our model applies to scenarios such as traffic planning, disaster evacuation, and other animal movements. 
Early arrivals secure better territories, while traveling together enhances navigation accuracy, foraging efficiency, and energy efficiency. Longer or more difficult migration routes reduce predation risks but increase travel costs, such as higher elevations and scarce food resources. Our analysis reveals a richer set of subgame perfect equilibria (SPEs) and heightened competition, compared to earlier models focused only on timing. By incorporating individual differences in travel costs, our model introduces a ``neutrality" state in addition to ``cooperation" and ``competition." 
\end{abstract}
%
\section{Introduction}

Groups form to provide collective benefits to their members, such as increased safety, shared resources, improved navigation accuracy, and greater energy efficiency\cite{sandler1980clubs, backstrom2006group, beauchamp2013social, beauchamp2021flocking}. Group formation is commonly observed not only in human society but also in nature. For example, people carpool to reduce travel costs, birds form flocks during migration to improve energy conservation and foraging efficiency \cite{mirzaeinia2019energy, beauchamp2011long, beauchamp2021flocking}, and cyclists ride in groups to lower energy expenditure. However, being part of a group also comes with costs, such as resource competition\cite{barker2012within}, coordination and communication efforts, and internal conflict. Resource competition arises because members must share collective resources according to a set of rules, which may not benefit everyone equally. \rev{In bird migration, while individuals in groups benefit from enhanced navigation accuracy and foraging/energy efficiency, they also compete for better territories on the breeding ground upon reaching their destination \cite{janiszewski2014timing}.} 

Formation of groups has been extensively studied across various disciplines, including economics, computer science, social science, political science, and biology \cite{couzin2006behavioral, goette2006impact, hollard2000existence, milchtaich2002stability}. Researchers have employed different approaches to examine how stable groups form, such as game theory \cite{banerjee2001core, bogomolnaia2002stability, slikker2001coalition, wang2024structural, kokko1999competition, janiszewski2014timing, wang2023cooperation, wang2024stackelberg}, clustering methods \cite{van2013community, wittemyer2005socioecology}, and agent-based modeling \cite{collins2017agent, collins2018strategic}. Classical game theoretical models capture the strategic decision-making of individuals, and could be used to understand interactions, predict outcomes, and guide the design of systems in competitive and cooperative settings. Most classical game-theoretic models found in this literature are one-shot games; relatively few are sequential games. 

Sequential game models bring their unique value for studying the timing of individual decisions in this context as they capture the dynamic nature of decision-making over time, unlike in a one-shot, simultaneous-move game. 
Our previous work \cite{wang2023cooperation} introduces \rev{a sequential game model that captures both competition and cooperation among individuals seeking to form groups to minimize traveling costs and predation risk.
Specifically, under this model, stronger agents make their decisions first, followed by weaker agents; each selects an arrival time at the destination, and they also compete for better territories upon arriving at the destination. With such a model we showed that when territorial differences are small, agents tend to form groups by choosing the same arrival time; otherwise, they compete by arriving as early as possible.
This is extended in \cite{wang2024stackelberg} by relaxing the definition of flocking to include all individuals arriving within a time window as part of the same group. This extension reveals an even richer set of flocking patterns at equilibrium.
The above models build upon Kokko's sequential but competition-only model \cite{kokko1999competition}, where the benefits of forming groups is not considered; as a result, at equilibrium, no group forms -- the strongest agent arrives first, followed by the next strongest, and so on, with the weakest arriving last. 
}

In this paper, we further extend the model in \cite{wang2023cooperation} by allowing agents to make a {\em route selection} decision in addition to timing. \rev{This extension is inspired by the commonly observed migration behavior, where migrants not only decide on timing to minimize costs such as those associated with seasonal climate but also choose a route, balancing the tradeoff between predation risk and challenges like high elevation and scarce food resources (easier routes are often prone to more predators)} \cite{alerstam2001detours, newton2023migration, egevang2010tracking}. To gain initial insight, we will limit ourselves to the simple setting of two agents, providing a framework to explore the decision-making differences of two types of individuals in migration. 

The remainder of this paper is organized as follows. Section~\ref{sec:flock formation game} introduces the Stackelberg game model. Section~\ref{sec:equilibrium} provides a detailed derivation of the Subgame Perfect Equilibria (SPEs), starting with the case of two routes, followed by the case of $n$ routes. 
Section~\ref{sec:discussion} compares the present model with the earlier game model in \cite{wang2023cooperation}. We also discuss how individual differences in travel costs affect route optimization. 
Section~\ref{sec:conclusion} concludes the paper.

\section{Game model and Preliminaries}
\label{sec:flock formation game}

Our game model follows closely the framework introduced in \cite{wang2023cooperation}, with the key difference that an agent makes an additional decision on a preferred route. \rev{As mentioned earlier, we will limit ourselves to the case of two agents.}

Consider two agents, denoted as $\mathcal{N}=\{1,2\}$, who must travel to a destination via one of $n$ available routes, $\mathcal{X}=\{1,2, ...,n\}$, to compete for territories at the destination, such as food resources and security. \rev{Traveling together (meaning to take the same route and to arrive at the same time)} provides several benefits, such as enhanced energy efficiency, increased safety, and improved navigation. On the other hand, arriving earlier can secure a better territory for the weaker agent -- if both arrive at the destination at the same time, the stronger agent will claim the better territory/resource.
Additionally, as the route becomes more challenging, travel cost rises, but predation risk decreases. Our main objective is to identify the conditions under which the two agents, despite their competitive goals, choose the same route and arrival time to travel together.

Each agent has a positive attribute called {\em strength}, denoted by $\beta_i >0$, $i\in\mathcal{N}$, which represents the agent's inherent qualities (e.g., navigation ability, survival skills, foraging ability, etc.). Similarly, each route is characterized by an attribute called {\em difficulty}, denoted by $\delta_k \geq 1$, $\cw{k} \in\mathcal{X}$, which reflects the challenges associated with the route. Each territory, or resource at the destination, is also associated with a positive value representing its quality, denoted by $E_j>0$, $j\in\mathcal{N}$ \rev{-- since each agent can only occupy a single territory, we will assume the same number of territories as there are agents.}
Without loss of generality, we will assume that Agent 1 is stronger than Agent 2: $\beta_1 > \beta_2$; routes are indexed in ascending order of their difficulties, route 1 being the easiest: 
$1 \leq \delta_1 < \delta_2< \cdots < \delta_n$; and that territory 1 is better than territory 2: $E_1 > E_2 > 0$. A direct consequence of the strength attribute is that the same journey with the same route selection is less costly for the stronger agent than it is for the weaker one. 

Agent $i$'s decision consists of its time of arrival at the destination and the selection of the route, denoted by the couple ($t_i, x_i$), \rev{where $t_i \in \mathcal{R}$ and $x_i\in{\cal{X}}$.} 
The joint action profile $(\mathbf{t, x}) = (t_1, x_1, t_2, x_2)$ or alternatively, $(\mathbf{t, x}) = (t_i, x_i, {t}_{-i}, {x}_{-i})$, following the notational convention. \rev{In this study, the term ``time'' in general refers to arrival time. If all routes take equal amount of time to travel, then arrival time is equivalent to departure time in terms of the agent's decision. If some routes take longer to travel, then a decision on arrival time implies the corresponding departure time.}

\subsection{Assumptions} 

\noindent (1) Full residency effect: The agent that arrives first and occupies a territory/resource at the destination will secure it. In other words, territories are claimed on a first-come, first-served basis.

\noindent (2) Tie-breaking: If both agents arrive simultaneously, the stronger agent will secure the better territory. 

\noindent (3) Sequential decision-making: Agent 1, the stronger agent and leader, first chooses its arrival time $t_1$ and route $x_1$, which is then announced. Following this, Agent 2, the follower, decides on its arrival time $t_2$ and route $x_2$. It is assumed that both agents commit to their decisions, meaning that Agent $i$ will travel along the selected route $x_i$ and arrive at the specified time $t_i$.
 
\subsection{The utility function} 
 The utility of agent $i$ is given by the following: 
\begin{equation}\label{eqn:u1} 
    u_i(\mathbf{t}, \mathbf{x}) = e_i(\mathbf{t}) - c_i(t_i, x_i) - p_i(\mathbf{t}, \mathbf{x})~, 
\end{equation}
where $e_i(\mathbf{t})$ represents the reward from the territory acquired by agent $i$, $c_i(t_i, \cw{x_i})$ denotes the travel cost, which depends solely on the agent's selected time and route, and $p_i(\mathbf{t})$ represents an additional cost determined by the chosen times and routes of all agents. This last term will also be referred to as the risk term. Each term is explained in detail below.

\paragraph{Benefit} if $t_i < t_{-i}$, $e_i(\mathbf{t}) = E_1$; if $t_i > t_{-i}$, $e_i(\mathbf{t}) = E_2$. When $t_i = t_{-i}$, the tie-breaking rule in the assumptions assigns $e_i(\mathbf{t}) = E_1$ to Agent 1 and $e_i(\mathbf{t}) = E_2$ to Agent 2. This is the only term that remains the same as in \cite{wang2023cooperation}.

\paragraph{Travel cost} $c_i(t_i, x_i) = \frac{1}{\beta_i}(t_i - t_o)^2+ c_o^i\delta(x_i) $, where $t_o$ is the optimal arrival time at which the travel cost is minimized to a marginal cost $c_o^i\delta(x_i)$. Deviations in either direction from $t_o$ increase the travel cost, and weaker agents (with smaller $\beta$) are more sensitive to this deviation from optimality. The term \rev{$\delta(x_i)=\delta_{x_i}$} represents the difficulty of the route chosen by the agent, 
where the more difficult the route, the higher the cost associated with traveling. For analytical convenience, we define $\delta_k$ \rev{as a sequence of increasing constants $\delta_k= \lambda_k\delta_1$, where $\lambda_1=1, \lambda_{k+1}>\lambda_k, k\geq 1$}. 
This model captures the travel cost due to external factors such as climate; e.g., in the case of bird migration during spring,  traveling in either excessively cold or hot weather can be detrimental. It also accounts for different routes, where routes with less exposure to predators may also expose migrants to extreme conditions, such as high elevation and scarce food. The fixed cost ($c_o^i>0$) is agent-dependent and generally lower for stronger agents (i.e., $c_o^1 < c_o^2$). 
In section~\ref{sec:equilibrium}, we will primarily focus on the impact of route selection on marginal travel cost rather than individual differences, and thus will assume a uniform marginal cost for each individual, i.e., $c_o^1 = c_o^2 = c_o $. We discuss the general case ($c_o^1 < c_o^2$) in Section~\ref{sec:discussion}.

We note that in \cite{wang2023cooperation}, since the agent's decision-making depends entirely on the optimal $t_o$, the marginal cost does not affect the analysis and is therefore set to zero ($c_o^i = 0$). However, in the current model, even if an agent selects the optimal travel time $t_o$ to minimize timing-related costs, a more challenging route will incur a higher travel cost. As a result, the term $c_o^i$ cannot be ignored.

 \paragraph{Risk} $p_i(\mathbf{t}, \mathbf{x}) = \frac{r}{|\{j:(t_j=t_i) \wedge (\delta(x_j) = \delta(x_i))\}|\delta(x_i)}$, where $|{j : (t_j = t_i) \wedge (\delta(x_j) = \delta(x_i))}|$ represents the total number of agents traveling on route $x_i$ and arriving at time $t_i$, including agent $i$ itself. Here, $r$ is the nominal risk for an individual agent. In this 2-agent scenario, the risk definition implies that if the two agents choose the same route and time of arrival, they are considered as traveling together, and the risk for each is reduced by half. Further, selecting a more difficult route also reduces this risk; e.g., predators are less likely to hunt along more challenging routes. This distinction is a key difference between the present model and our prior work \cite{wang2023cooperation}, which did not consider route selection.

 The solution concept we will use is the subgame perfect equilibrium (SPE) -- a strategy profile in which agents' strategies form a Nash equilibrium in every subgame of the original game -- with the specified ordering of decision making by the two agents; \rev{this will generally by denoted by $(\mathbf{t}^{*}, \mathbf{x}^{*})$.} In this analysis, agents face the challenge of selecting the optimal route and arrival time, while considering both individual risks and the strategic advantages of cooperating with another. By examining the interactions between these factors, we aim to gain a deeper understanding of how agents navigate complex decision-making scenarios, balancing competition and collaboration, and optimizing their choices in environments with limited resources.

\section{Equilibrium Analysis}\label{sec:equilibrium}
In this section we first study the SPE of the 2-agent 2-route game, and then generalize the result to the 2-agent $n$-route game. \cw{ For analytical tractability, we assume that both agents use the same deterministic rule to break ties whenever multiple routes are equally optimal.}

\subsection{The 2-agent 2-route game}
\rev{With only two routes, the route difficulty is given by $\delta_1$ and $\delta_2 = \lambda \delta_1$, $\lambda > 1$, respectively.}

\begin{prop}\label{prop:t0}
No agent arrives later than $t_o$ in an SPE. 
\end{prop} 

\noindent{\em Proof Sketch:} 
This can be shown by contradiction. Consider the later arrival time between the two agents, and denote this by $t>t_o$.  
\begin{itemize} 
\item If agent 2 arrives alone at $t$, it can advance its arrival time by an infinitesimal amount $\epsilon$ to lower its travel cost while keeping the same route, thereby improving its utility. Thus, $t$ cannot be part of an SPE. 
\item If agent 1 arrives alone at $t$ (meaning agent 2 chose an earlier time), advancing this by $\epsilon$ either has no impact on agent 2’s decision or causes agent 2 to delay its arrival and join agent 1 to form a group on the same route. In either case, agent 1's utility improves, meaning $t$ is not its best response, and cannot be part of an SPE. 
\item If both agents arrive at $t$, then agent 1 can advance its arrival by $\epsilon$, causing agent 2 to do the same as it lowers its travel cost, which in turn improves agent 1's utility. Thus $t$ is not an optimal choice for agent 1 and cannot be part of an SPE. \end{itemize}

Proposition~\ref{prop:t0} is similar to Proposition 1 in \cite{wang2023cooperation} because while route selection may affect departure time, it does not impact arrival decisions.

\begin{prop}\label{prop:x1=x2, t1=t2}
In any SPE where $t_1^* = t_2^*$, it must follow that $x_1^* = x_2^*$. In other words, whenever agents do not compete for resources in equilibrium, they will also choose to travel together on the same route.
\end{prop}

\noindent{\em Proof Sketch:} 
\rev{This can also be done by contradiction. By choosing to arrive at the same time but on a different route, agent 2 will take the lower-quality territory without benefiting from flocking. If agent 2 chooses to do so because it results in a higher utility, then it can improve it even more by moving up its arrival by an infinitesimal amount $\epsilon$, which has virtually no impact on its cost of travel but allows it to take the better territory. This shows that $t_2^*=t_1^{*}, x_2^{*}\neq x_1^{*}$ is not a best response by agent 2.}

\begin{prop}\label{prop:x1=x2, t1 neq t2}
\rev{There exists no SPE where $t_1^* \neq t_2^*$ and $x_1^* \neq x_2^*$.}  
\end{prop}

\noindent{\em Proof Sketch:} 
\rev{Since neither agent benefits from traveling together, then given $t_1^* \neq t_2^*$ and the marginal travel cost $c_o$ is the same for both, their optimal route choice is completely determined by the route difficulty and is independent of the individual strength. Thus the optimal route for one agent must also be optimal for the other, meaning we must have $x^*_1=x^*_2$ uder the same tie-breaking rule.}

Proposition~\ref{prop:x1=x2, t1 neq t2} implies 
when agents compete for resources in an SPE ($t_1^* \neq t_2^*$), they will choose the same route but not travel together.
The following corollary follows directly from Proposition~\ref{prop:x1=x2, t1=t2} and Proposition~\ref{prop:x1=x2, t1 neq t2}.
\begin{coro}\label{coro:x1=x2 all}
Both agents will select the same route in any SPE: $x_1^* = x_2^*$.
\end{coro}

\begin{prop}\label{coro:t1 < = t2 = t0}
The weaker agent will 
\rev{always arrive at $t_o$}
and the stronger agent will arrive no later than the weaker agent in any SPE: $t_1^* \leq t_2^* = t_o$.
\end{prop}

\noindent{\em Proof Sketch:}
\rev{From Proposition \ref{prop:t0}, we know if arrival times are different from $t_o$, then it must be earlier than $t_o$. We proceed with contradiction again.\\
(1) Suppose $t_2^*<t_1^*\leq t_o$, then agent 2 can always delay its arrival by an infinitesimal amount $\epsilon$ without losing the better territory while lowering its travel cost. Therefore, agent 2 can never arrive before agent 1 in an SPE.\\
(2) Suppose $t_1^*<t_2^*<t_o$, then agent 2 is not benefiting from flocking, thus it can always delay its arrival to $t_o$ to minimize its travel cost. Therefore this cannot be an SPE.\\
(3) Suppose $t_1^* = t_2^* < t_o$. From Proposition~\ref{prop:x1=x2, t1=t2} we know $x_1^* = x_2^*$. This suggests that agent 2 finds flocking with agent 1 more beneficial than traveling alone (even earlier) and taking the better territory. It follows that agent 1 can delay its arrival by an infinitesimal amount $\epsilon$ without causing agent 2 to leave the flock -- its utility has improved, but too small to trigger a decision to travel alone. Agent 1's utility would also improve, therefore this cannot be an SPE.}

\begin{theo}\label{theo:all SPE}
The pure strategy SPEs for the 2-agent 2-route game has to take the follow form: $(\mathbf{t}^*, \mathbf{x}^*) = 
    (t_1^*, x^*, t_o, x^*)$, where $t_1^* \in \{T, t_o\}$, $x^* \in \{1,2 \}$, and $T = t_o - \sqrt{(E_1-E_2)\beta_2}$. They occur in the following cases:\\
(1) If $c_o \delta_1^2 < \frac{r}{2\lambda}$:
\begin{itemize}
    \item If $E_1-E_2 \leq \frac{r}{2\lambda \delta_1}$, then $(
     \mathbf{t}^*, \mathbf{x}^*) = 
    (t_o, 2, t_o, 2)$ and is unique;
    \item Otherwise, $(\mathbf{t}^*, \mathbf{x}^*) = 
    (T, 2, t_o, 2)$ is the unique SPE;
\end{itemize} 
(2) If $c_o \delta_1^2 = \frac{r}{2\lambda}$:
\begin{itemize}
    \item If $E_1-E_2 \leq \frac{r}{2\lambda \delta_1}$, then 
    $(t_o, 2, t_o, 2)$ and $(t_o, 1, t_o, 1)$ are both SPEs;
    \item Otherwise, 
    $(T, 2, t_o, 2)$ is the unique SPE;
\end{itemize}
(3) If $c_o \delta_1^2 \in (\frac{r}{2\lambda}, \frac{r}{\lambda})$:
\begin{itemize}
    \item If $E_1-E_2 \leq (\lambda-1)c_o\delta_1 - \frac{(\lambda - 2)r}{2 \lambda \delta_1}$, then 
    $(t_o, 1, t_o, 1)$ is the unique SPE;
    \item Otherwise,  
    $(T, 2, t_o, 2)$ is the unique SPE;
\end{itemize}
(4) If $c_o \delta_1^2 > \frac{r}{\lambda}$:
\begin{itemize}
    \item If $E_1-E_2 \leq \frac{r}{2 \delta_1}$, then $(t_o, 1, t_o, 1)$ is the unique SPE;
    \item Otherwise, $(T, 1, t_o, 1)$ is the unique SPE;
\end{itemize} 
(5) If $c_o \delta_1^2 = \frac{r}{\lambda}$:
\begin{itemize}
    \item If $E_1-E_2 \leq \frac{r}{2 \delta_1}$, then 
    $(t_o, 1, t_o, 1)$ is the unique SPE;
    \item Otherwise, $(T, 1, t_o, 1)$ and $(T, 2, t_o, 2)$ are both SPEs;
    \end{itemize} 

\end{theo}


These cases are depicted in Figure \ref{fig:4MFFG}. 
We briefly explain each case in Theorem~\ref{theo:all SPE} below and highlight several insights:
\begin{itemize}
\item Case (1): The condition of this case suggests the reduction in risk provided by the harder route surpasses the additional travel cost ($c_o \delta_1^2 < \frac{r}{2\lambda}$); therefore both agents tend to choose the harder route, regardless of whether they are competing. In this situation, their decision to compete rests solely on evaluating the resource or territory differences against the benefits of flocking.

\item Case (2): This is the boundary case of (1), where the two routes \rev{are essentially equivalent if shared}. As a result, if they are not competing for the better territory, Agent 1 is indifferent between the two routes, while Agent 2 will follow Agent 1 and choose the same route.

\item Case (3): 
This is a case where the agents favor the harder route if they are competing, and the easier route if they are not. Thus, their decision to compete is influenced not only by the balance between territory differences and the potential benefits of forming a group, but also by their choice of route. 

\item Case (4): The condition of this case suggests that the savings in travel cost of the easier route outweigh the extra risk it entails ($c_o \delta_1^2 > \frac{r}{\lambda}$); therefore, both agents are inclined to select the easier route, irrespective of their competitive status. Here, their decision to compete is based on weighing the territory difference against the potential advantages of flocking.

\item Case (5): This is the boundary case of (4), where the two routes \rev{are essentially equivalent if not shared. As a result, if the two agents are competing, Agent 1 chooses the easier route and is joined by agent 2; otherwise, if traveling alone, they are indifferent between the two routes.}

\end{itemize}

\begin{figure}[ht!]
\centerline{\includegraphics[width=0.99\linewidth]{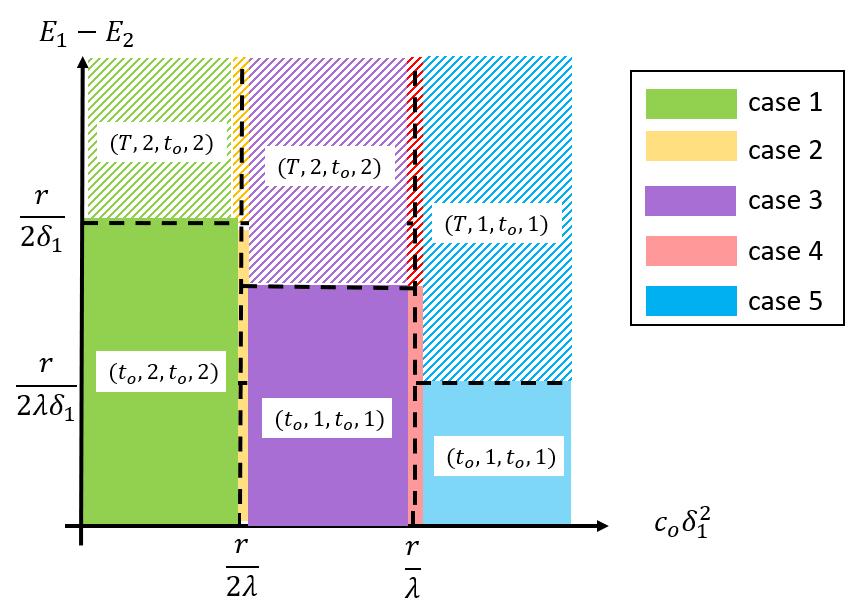}}
\caption{ 
A depiction of the different cases outlined in Theorem~\ref{theo:all SPE} and the corresponding SPEs.
The $x$-axis is the first condition in each case and represents the tradeoff between travel cost and risk, which is captured by the relationship between $c_o\delta_1^2$ and $\frac{r}{\lambda}$. The $y$-axis is the second condition in each case and represents the difference between the quality of the two territories $E_1 - E_2$. When $c_o \delta_1^2 \in (\frac{r}{2\lambda}, \frac{r}{\lambda})$, the horizontal line refers to \rev{$E_1-E_2 = (\lambda-1)c_o\delta_1 - \frac{(\lambda - 2)r}{2 \lambda \delta_1}$}. Solid colors represent SPEs without competition, patterned colors with competition. For cases 2 and 4, the corresponding SPEs include those in adjacent areas on both sides (solid and patterned).}
\label{fig:4MFFG}
\end{figure}

\subsection{The 2-agent $n$-route game}

The 2-route game can be generalized to an $n$-route game in a straightforward manner. By following arguments very similar to Corollary~\ref{coro:x1=x2 all} and Proposition~\ref{coro:t1 < = t2 = t0}, we have the following result. 
\begin{coro}
    In any SPE of a 2-agent n-route game, $t_1^*\leq t_2^* = t_o$ and $x_1^* = x_2^*$.
\end{coro}

Similarly, the optimal route selection for both agents depends on the relationship between the associated travel cost and risk. Since a more challenging route offers lower risk but incurs higher travel costs, any route that minimizes the combination of these two, given $c_o$ and $r$, \rev{will be chosen by Agent 1 when traveling alone. And subsequently agent 2 will do exactly the same.}
Additionally, when the two agents compete for territories ($t_1^* = T$ and $t_2^* = t_o$), the optimal arrival time for agent 1  is independent of the route choice. Therefore, the optimal route choice is given by $x^* = \underset{x \in \mathcal{X}}{\arg\min} \, (c_o \delta(x) + \frac{r}{\delta(x)})$. If they do not compete ($t_1^* = t_2^* = t_o$), the optimal route choice is given by $x^* = \underset{x \in \mathcal{X}}{\arg\min} \, (c_o \delta(x) + \frac{r}{2\delta(x)})$ \rev{on account of benefits from flocking}. \rev{Consequently, with suitable tie-breaking, only 2 routes out of $n$ are of relevance and appear in any SPE.} 

\begin{coro}
   \rev{The pure strategy SPEs of the 2-agent $n$-route game are of the following types:}\\
   (1) $(t_o, x^*, t_o, x^*)$, where $x^* = \underset{x \in \mathcal{X}}{\arg\min} \, (c_o \delta(x) + \frac{r}{2\delta(x)})$;\\
   (2) $(T, x^*, t_o, x^*)$, where $T = t_o-\sqrt{(E_1-E_2)\beta_2}$ and $x^* = \underset{x \in \mathcal{X}}{\arg\min} \, (c_o \delta(x) + \frac{r}{\delta(x)})$. 
\end{coro}

\section{Discussion}\label{sec:discussion}
We first compare our model to the earlier game model without route selection. Examined in the 2-agent case, we show that the introduction of route difficulty can increase competition between agents. Furthermore, accounting for individual differences in the marginal travel costs significantly impacts route optimization and equilibrium outcomes, resulting in an enriched set of SPEs. In particular, one new type of SPE, referred to as a ``neutrality" state, emerges when agents arrive synchronously but on different routes. We end this section by connecting our model closer to the bird migration context and offering interpretations of the different elements of this model in this context.

\subsection{Comparison with the earlier flock formation game  \cite{wang2023cooperation}}

If we reduce the route choices to only one, then the current game reduces to the earlier basic flock formation game in \cite{wang2023cooperation}. \rev{In particular, the following follows directly from Theorem \ref{theo:all SPE}:} 
\begin{coro}
   There exist exactly two types of pure strategy SPEs for the 2-agent 1-route game:\\
   (1) If $E_1 - E_2 \leq \frac{r}{2\delta_1}$, then  $(t_o, 1, t_o, 1)$ is the unique SPE;\\
   (2) If $E_1 - E_2 > \frac{r}{2\delta_1}$, then  $(T, 1, t_o, 1)$, $T = t_o-\sqrt{(E_1-E_2)\beta_2}$ is the unique SPE.
\end{coro}
Setting $\delta_1 =1$ gives the version obtained in \cite{wang2023cooperation}. We additionally note that if $\delta_1 >1$, then $\frac{r}{2\delta_1} < \frac{r}{2}$, indicating that competition between the two agents is more intense \rev{(the parameter range in which they compete is larger) when the route becomes harder relative to other constants in the model.  This is understandable as a harder route increases travel cost while decreasing the predation risk, thereby reducing the benefit of flocking.} 

\subsection{The significance of individual strength and marginal travel cost}\label{subsec: extension}
Our analysis so far is based on the simplification of equal marginal travel cost between the two agents: $c_o^1 = c_o^2 = c_o$, and our main observation is that the agents will either (fully) cooperate (same time, same route) or compete (different time, same route). The fact that they always prefer the same route (under suitable tie-breaking) is rooted in this simplifying assumption that ignores individual differences in marginal travel costs on the same route. In reality this individual difference does play a role, e.g., in bird migration, a harder route may be more costly for weaker individuals, such as the young and/or female, than for stronger individuals, such as adult male.

A natural extension is to adopt the assumption that $c_0^1 < c_o^2$. Consequently, Corollary~\ref{coro:x1=x2 all} is no longer valid in general, as the optimal route choice for one agent may differ from that for another. Consider for instance the scenario where $t_1^* < t_2^* = t_o$. In such cases, if $c_o^i \delta_1^2 \geq \frac{r}{\lambda}$, then $x_i^* = 1$; conversely, if $c_o^i \delta_1^2 < \frac{r}{\lambda}$, then $x_i^* = 2$. Given $c_0^1 < c_o^2$, a situation may arise where $c_o^1 \delta_1^2 < \frac{r}{\lambda} \leq c_o^2 \delta_1^2$, resulting in $x_1^* = 2$ and  $x_2^* = 1$. Similar conditions exist for $t_1^* = t_2^* = t_o$ where $x_1^* \neq x_2^*$. \rev{We will refer to these new types of SPEs as a state of \textit{neutrality}, in contrast with both {\em competition} and {\em cooperation}, the latter two of which we have been using to describe different SPEs.} In this type of SPEs, agents exhibit independence in route selection while maintaining coordination in arrival times, which determines their territory/reward. 
This state adds an interesting complexity to the equilibrium analysis as agents independently optimize their routes based on individual cost-benefit analyses, resulting in stable yet uncoordinated route choices that do not interfere with the overall timing coordination.

\rev{With this extension, two new SPEs are introduced, resulting in the complete set of SPEs given by: 
$(T, x^*, t_o, x^*)$, $(T, x^*, t_o, x^{-*})$, $(t_o, x^*, t_o, x^*)$, and $(t_o, x^*, t_o, x^{-*})$,  where $x^{-*} = \{1, 2\}-x^{*}$.} 

This expanded set of SPEs reflects the diversity of strategic arrangements arising from both the coordination of timing and the independent route selections based on individual agent preferences. By accounting for these individual differences, the model illustrates how agents can reach different equilibria that accommodate their unique cost-benefit analyses, leading to varied strategic outcomes.

\subsection{Observations and interpretations in the context of bird migration}
In bird migration, the timing and route selection are crucial decisions that seasonal migrants must carefully consider \cite{alerstam2001detours, newton2023migration}. Routes are typically chosen based on factors such as predation risk, energy efficiency, and food resources, with avian species often relying on social networks such as their flocks to guide their decisions and reduce travel cost or predation risk \cite{egevang2010tracking, mckinnon2010lower}. The timing of the journey is similarly influenced by these factors, as many migrants aim to minimize risks and maximize opportunities by aligning their travel with seasonal changes. \rev{The formation of groups during migration is commonly observed \cite{mirzaeinia2019energy, beauchamp2011long, beauchamp2021flocking, janiszewski2014timing}, with strong evidence that} it increases energy efficiency and navigation accuracy, and decreases predation risk. The dynamics of flock formation highlight the complex interplay between individual choices and environmental influences during migration. This interplay is further influenced by age or sex-specific factors, which play a significant role in shaping migration patterns and decisions. 

Differences in route selection between young and adult male birds during spring migration are largely influenced by variations in experience, energy reserves, and competitive pressures \cite{verhoeven2022age, mckinnon2014tracking}. Young males tend to take longer, less direct routes with more stopovers, often avoiding high-competition areas to compensate for lower fat reserves and inexperience. In contrast, adult males opt for shorter, more direct routes with fewer stopovers, prioritizing early arrival at breeding grounds and relying on learned migration strategies. 
Furthermore, strong patterns of sex-based differential migration \rev{has also been observed} in many avian species \cite{mckellar2025sex}. 
Our 2-agent $n$-route game could be used to capture this difference in decisions about departure timing (determined by time of arrival and route selection). By considering adult males (or males) as the stronger agents and juvenile males (or females) as the weaker agents based on average strength, our model takes into account the travel costs and risks of individual differences. It explores the dilemma between choosing a more difficult or longer route to avoid predation and the potential costs of such a route. \rev{Additional individual differences in migratory routes and behaviors are in general less common than age and sex-based differences.}  However, increasing effects of climate change may stimulate greater individual variation in route selection among migratory birds; therefore, our extended model discussed in Section~\ref{subsec: extension} can be particularly useful for capturing and understanding these dynamics.

\rev{It also bears mentioning that the simpler model which assumes a uniform marginal cost, and which results in agents favoring the same route (Section \ref{sec:equilibrium}), also offers insight and explanation} of the more commonly observed phenomenon
in this literature: many avian migrants, regardless of sex or age, typically choose routes based on experience with established paths or those recommended by social networks, or travel in mixed-species groups, leading to shared route patterns. 

\section{Conclusion}\label{sec:conclusion} 
We introduced a Stackelberg game to study how rational agents make strategic decisions on timing and route selection in order to reach a destination and collect a reward.  They benefit from group formation if they travel at the same time (measured by simultaneous arrival) and on the same route. We analyzed the properties of the SPEs of this game, highlighting a richer set of SPE types with intensified competition compared to earlier work that focused solely on timing without considering route selection.

The main analysis is centered on the 2-agent 2-route model, but the results are easily generalized to the 2-agent n-route game. We first showed that when the difficulty of a route is perceived the same by both agents (the same marginal travel cost), then in any SPE, the agents share the same route preference. 
We then extended our model by considering individual differences in marginal travel costs, and showed that an additional ``neutrality" state may emerge in equilibrium, where agents arrive simultaneously but choose separate routes. 
Both models find support in the bird migration literature. An important next step would be to incorporate climate-driven changes in route conditions and agent decision-making into the model, thereby reflecting the greater individual variation seen as migratory environments shift.

\nocite{*}
\bibliographystyle{unsrt}
\bibliography{myreference}

\appendix
The following is the full proofs for Proposition~\ref{prop:x1=x2, t1=t2}, Proposition~\ref{prop:x1=x2, t1 neq t2}, and Proposition~\ref{coro:t1 < = t2 = t0} respectively.

\begin{proof}[Proof of Proposition~\ref{prop:x1=x2, t1=t2}]
    This can be shown by contradiction. Let's assume that $\mathbf{t^*, x^*} = (t_1^*, x_1^*, t_2^*, x_2^*)$, where $t_1^* = t_2^* = t^*$ and $x_1^* \neq x_2^*$ is an SPE. This leads us to two potential scenarios for route selections: (1) $x_1^* = 1$, $x_2^* = 2$, (2) $x_1^* = 2$, $ x_2^* = 1$. Let's denote $\frac{1}{\beta_i}(t^* - t_o)^2$ as $c(t^*)$.
    
    {\bf Case 1: $x_1^* = 1$, $x_2^* = 2$}
    
    Based on our assumption, for agent 2: $\forall x_2 \in \{1,2\}$, $u_2(t^*, x_1^*, t^*, x_2^*) \geq u_2(t^*, x_1^*, t^*, x_2)$. 
    \begin{equation}
        u_2(t^*, x_1^* = 1, t^*, x_2^* = 2) = E_2 - c(t^*)- c_o\lambda\delta_1 - \frac{r}{\lambda\delta_1}
    \end{equation}
        \begin{equation}
        u_2(t^*, x_1^* = 1, t^*, x_2 = 1) = E_2 -  c(t^*)-c_o\delta_1 - \frac{r}{2\delta_1}
    \end{equation}
    Therefore, we find 
        \begin{equation}\label{eq:condition1}
            \frac{\delta_1^2c_o}{r} \leq \frac{\lambda -2}{2\lambda (\lambda -1)}
        \end{equation} 
    Given that $\frac{\delta_1^2c_o}{r} > 0$, it implies $\lambda > 2$.
    If agent 1 switches to the other route (i.e., $x_1 = 2$), we can find the best response for agent 2 by comparing the following:
             \begin{equation}
                 u_2(t^*, x_1 = 2, t^*, x_2 = 1) = E_2 -   c(t^*) - c_o\delta_1 - \frac{r}{\delta_1} 
             \end{equation}
            \begin{equation}
            u_2(t^*, x_1 = 2, t^*, x_2 = 2) = E_2 -  c(t^*)- c_o\lambda\delta_1 - \frac{r}{2\lambda\delta_1} \end{equation}
    We have $u_2(t^*, x_1 = 2, t^*, x_2 = 1) - u_2(t^*, x_1 = 2, t^*, x_2 = 2) = \frac{(\lambda -1)r}{\delta_1} (\frac{\delta_1^2c_o}{r} - \frac{2\lambda -1}{2\lambda (\lambda -1)})$. Since $2\lambda - 1 > \lambda -2$ and the condition in \eqref{eq:condition1}, we find $u_2(t^*, x_1 = 2, t^*, x_2 = 1) - u_2(t^*, x_1 = 2, t^*, x_2 = 2) < 0$. Therefore, $x_2^*(t^*, x_1 = 2, t^*) = 2$.
    Next, we check whether agent 1 has an incentive to deviate to the other route, given agent 2'best responses:
            \begin{equation}
            u_1(t^*, x_1^* = 1, t^*, x_2^* = 2) = E_1-  c(t^*) - c_o\delta_1 - \frac{r}{\delta_1} 
            \end{equation}
            \begin{equation}
            u_1(t^*, x_1 = 2, t^*, x_2^* = 2) = E_1 -  c(t^*)- c_o\lambda\delta_1 - \frac{r}{2\lambda\delta_1} \end{equation}
    We have $u_1(t^*, x_1 = 2, t^*, x_2^* = 2) - u_1(t^*, x_1^* = 1, t^*, x_2^* = 2) = \frac{(\lambda -1)r}{\delta_1} (\frac{2\lambda -1}{2\lambda (\lambda -1)} - \frac{\delta_1^2c_o}{r})$. Again since $2\lambda - 1 > \lambda -2$ and the condition given by \eqref{eq:condition1}, we find $u_1(t^*, x_1 = 2, t^*, x_2^* = 2) - u_1(t^*, x_1^* = 1, t^*, x_2^* = 2)>0$. In other words, agent 1 has an incentive to deviate to the other route, which implies that $\mathbf{t^*, x^*} = (\cw{t^*}, x_1^*=1, \cw{t^*}, x_2^*=2)$ is not an SPE.

        {\bf Case 2: $x_1^* = 2$, $x_2^* = 1$}
    
    Based on our assumption, for agent 2: $\forall x_2 \in \{1,2\}$, $u_2(t^*, x_1^*, t^*, x_2^*) \geq u_2(t^*, x_1^*, t^*, x_2)$. 
           \begin{equation}
                u_2(t^*, x_1^* = 2, t^*, x_2^* = 1) = E_2 -  c(t^*)- c_o\delta_1 - \frac{r}{\delta_1}
           \end{equation}
        \begin{equation}
        u_2(t^*, x_1^* = 2, t^*, x_2 = 2) = E_2 -  c(t^*)- c_o\lambda\delta_1 - \frac{r}{2\lambda\delta_1} 
        \end{equation}
    Therefore, we find 
        \begin{equation}\label{eq:condition2}
            \frac{\delta_1^2c_o}{r} \geq \frac{2\lambda -1}{2\lambda (\lambda -1)}
        \end{equation} 
    If agent 1 switches to the other route (i.e., $x_1 = 1$), we can find the best response for agent 2 by comparing the following:
            \begin{equation}
            u_2(t^*, x_1 = 1, t^*, x_2 = 1) = E_2-  c(t^*) - c_o\delta_1 - \frac{r}{2\delta_1} 
            \end{equation}
            \begin{equation}
            u_2(t^*, x_1 = 1, t^*, x_2 = 2) = E_2-  c(t^*) - c_o\lambda\delta_1 - \frac{r}{\lambda\delta_1} 
            \end{equation}
    We have $u_2(t^*, x_1 = 1, t^*, x_2 = 1) - u_2(t^*, x_1 = 1, t^*, x_2 = 2) = \frac{(\lambda -1)r}{\delta_1} (\frac{\delta_1^2c_o}{r} - \frac{\lambda -2}{2\lambda (\lambda -1)})$. Since $2\lambda - 1 > \lambda -2$ and the condition in \eqref{eq:condition2}, we find $u_2(t^*, x_1 = 1, t^*, x_2 = 1) - u_2(t^*, x_1 = 1, t^*, x_2 = 2) > 0$. Therefore, $x_2^*(t^*, x_1 = 1, t^*) = 1$.
    Next, we check whether agent 1 has an incentive to deviate to the other route, given agent 2'best responses:
            \begin{equation}
            u_1(t^*, x_1^* = 2, t^*, x_2^* = 1) = E_1 -  c(t^*)- c_o\lambda\delta_1 - \frac{r}{\lambda\delta_1} 
            \end{equation}
            \begin{equation}
            u_1(t^*, x_1 = 1, t^*, x_2^* = 1) = E_1 -  c(t^*)- c_o\delta_1 - \frac{r}{2\delta_1}  
            \end{equation}
    We have $u_1(t^*, x_1^* = 2, t^*, x_2^* = 1) -  u_1(t^*, x_1 = 1, t^*, x_2^* = 1) = \frac{(\lambda -1)r}{\delta_1} (\frac{\lambda -2}{2\lambda (\lambda -1)} - \frac{\delta_1^2c_o}{r})$. Again since $2\lambda - 1 > \lambda -2$ and the condition given by \eqref{eq:condition2}, we find $u_1(t^*, x_1^* = 2, t^*, x_2^* = 1) - u_1(t^*, x_1 = 1, t^*, x_2^* = 1) < 0$. In other words, agent 1 has an incentive to deviate to the other route, which implies that $\mathbf{t^*, x^*} = (t^*, x_1^*=2, \cw{t^*}, x_2^*=1)$ is not an SPE.
\end{proof}

\begin{proof}[Proof of Proposition~\ref{prop:x1=x2, t1 neq t2}]
This can be shown by contradiction.  Let's assume that $\mathbf{t^*, x^*} = (t_1^*, x_1^*, t_2^*, x_2^*)$, where $t_1^* \neq t_2^* $ and $x_1^* \neq x_2^*$ is an SPE. Therefore, for agent 2: $\forall x_2 \in \{1,2\}$, $u_2(t_1^*, x_1^*, t_2^*, x_2^*) \geq u_2(t_1^*, x_1^*, t_2^*, x_2)$. Since $x_1^* \neq x_2^*$, we have $x_2 \in \{x_1^*, x_2^*\}$. Consequently, $u_2(t_1^*, x_1^*, t_2^*, x_2^*) \geq u_2(t_1^*, x_1^*, t_2^*, x_1^*)$. Define $\Delta u_2 = u_2(t_1^*, x_1^*, t_2^*, x_2^*) - u_2(t_1^*, x_1^*, t_2^*, x_1^*)$, leading $\Delta u_2 \geq 0$. Thus, 
\begin{equation}
    \Delta u_2 = (\delta(x_1^*) -\delta (x_2^*))(c_o - \frac{r}{\delta(x_2^*)\delta(x_1^*)}) \geq 0
\end{equation} 
Next, we examine the utility difference for agent 1 between choosing $x_2^*$ and $x_1^*$. Define $\Delta u_1 = u_1(t_1^*, x_2^*, t_2^*, x_2) - u_1(t_1^*, x_1^*, t_2^*, \hat{x}_2)$, where $x_2, \hat{x}_2 \in \{x_1^*, x_2^*\}$. Then, for all 
$x_2, \hat{x}_2 \in \{x_1^*, x_2^*\}$, we have
\begin{equation}
    \Delta u_1 = (\delta(x_1^*) -\delta (x_2^*))(c_o - \frac{r}{\delta(x_2^*)\delta(x_1^*)})
\end{equation}
Therefore, $\Delta u_1 = \Delta u_2$, and thus $\Delta u_1 \geq 0$. Thus, let $\hat{x}_2 = x_2^*$, we have $u_1(t_1^*, x_2^*, t_2^*, x_2) - u_1(t_1^*, x_1^*, t_2^*, x_2^*) \geq 0$. Therefore, $(t_1^*, x_1^*, t_2^*, x_2^*)$ is not an SPE.
\end{proof}

\begin{proof}[Proof of Proposition~\ref{coro:t1 < = t2 = t0}]
    This can be shown by two separate cases: \\
    (1) in any SPE where $t_1^* = t_2^*$, Proposition~\ref{prop:x1=x2, t1=t2} demonstrates $x_1^* = x_2^*$.  it is evident that the optimal arrival time for both players is $t_o$, in order to minimize their respective travel costs.  Therefore, we conclude that $t_1^* = t_2^*=t_o$; \\
    (2) in any SPE where $t_1^* \neq t_2^*$, First, we can show that $\max {(t_1^*, t_2^*)} = t_o$. Next, we show $t_1^* < t_2^*$ using contradiction. 
    \begin{enumerate}
        \item This can be shown using contradiction. Let's assume that the one that arrives later will not arrive at $t_o$ (i.e. $\max {(t_1^*, t_2^*)} \neq t_o$). From Proposition~\ref{prop:t0}, we have $\max {(t_1^*, t_2^*)} \leq t_o$. Therefore, our assumption is $\max {(t_1^*, t_2^*)} < t_o$. From Proposition~\ref{prop:x1=x2, t1 neq t2}, we know $x_1^* = x_2^*$. Let $x_1^* = x_2^* = x^*$. \\
(a) If agent 2 is the one who arrives later (i.e. $t_1^* < t_2^*<t_o$), then $u_2(t_1^*, x^*, t_2^*, x^*) - u_2(t_1^*, x^*, t_o, x^*) = -\frac{(t_2^*-t_o)^2}{\beta_2} < 0$. It indicates that agent 2 will deviate by arriving later at $t_o$, and it contradicts the fact that $(t_1^*, x_1^*, t_2^*, x_2^*)$ is an SPE. Therefore, $t_2^* = t_o$;\\
(b) If agent 1 is the one who arrives later (i.e., $t_2^* < t_1^*<t_o$), then if agent 1 chooses to arrive later at $t_o$, agent 2 will either arrive later or arrive at $t_o$. Note that the route selection may only change if both choose to arrive at the same time. However, this will reduce the risk by half and thus make the favorable route even better. Therefore, the route selection will stay the same. For either situation, we can find the utility of agent 1 increases: (a) $u_1(t_o, x^*, t_2 \in (t_2^*, t_o), x^*) - u_1(t_1^*, x^*, t_2^*, x^*) = \frac{(t_1^*-t_o)^2}{\beta_1} > 0$; (b) $u_1(t_o, x^*, t_o, x^*) - u_1(t_1^*, x^*, t_2^*, x^*) = \frac{(t_1^*-t_o)^2}{\beta_1} + \frac{r}{2 \delta(x^*)} > 0$. Thus, it contradicts with the fact that $(t_1^*, x_1^*, t_2^*, x_2^*)$ is an SPE. Therefore, $t_1^* = t_o$.\\
\cw{Therefore, combining (a) and (b), we have shown $\max {(t_1^*, t_2^*)} = t_o$.}
        \item Assume $ t_2^* < t_1^* = t_o$. For agent 2, it can always arrive slightly later at $t_2^* +\epsilon $ with $\epsilon \in (0, t_o-t_2^*)$ to decrease its traveling cost while securing the same resource and keeping the same risk. Therefore $t_2^* < t_o$ is not the best response for agent 2, and thus it is not an SPE. As a result, we have $t_1^* < t_2^* = t_o$.
    \end{enumerate}
\end{proof}

\begin{proof} [Proof of Theorem~\ref{theo:all SPE}]
By Corollary~\ref{coro:x1=x2 all} and Proposition~\ref{coro:t1 < = t2 = t0},  it follows that in any SPE, we have: $t_1^* \leq t_2^* = t_o$ and $x_1^* = x_2^*$. Let $x_1^* = x_2^* = x^*$.
We can begin by identifying the conditions under which both agents will prefer one route over the other in the following scenarios:
\begin{enumerate}
    \item When there is cooperation but no competition ($t_1^* = t_2^* = t_o$):
    \begin{itemize}
        \item If $c_o \delta_1^2 > \frac{r}{2\lambda}$, then $x^* = 1$;
        \item If $c_o \delta_1^2 = \frac{r}{2\lambda}$, then $x^* = 1$ or $x^* = 2$;
        \item If $c_o \delta_1^2 < \frac{r}{2\lambda}$, then $x^* = 2$;
    \end{itemize}
    \item When there is competition but no cooperation ($t_1^* < t_2^* = t_o$):
    \begin{itemize}
        \item If $c_o \delta_1^2 > \frac{r}{\lambda}$, then $x^* = 1$;
        \item If $c_o \delta_1^2 = \frac{r}{\lambda}$, then $x^* = 1$ or $x^* = 2$;
        \item If $c_o \delta_1^2 < \frac{r}{\lambda}$, then $x^* = 2$;
    \end{itemize}
\end{enumerate}
Therefore, by combining both scenarios, we can identify five distinct sets of conditions that determine route selection preferences.

Next, if agent 2 competes with agent 1 in any SPE, we need to determine the best arrival time for agent 1. Since the optimal route selection is always the same for both agents, the optimal arrival time is similar as that shown in \cite{wang2023cooperation}. It is the tipping time where agent 2 is indifferent between arriving alone at $t_o$ and arriving first at $T$:
\begin{equation}\label{eq:2collaborate}
        u_2(t_o, x^*, t_o, x^*) = E_2 - c_o\delta(x^*) - \frac{r}{\delta(x^*)}
    \end{equation}
    \begin{equation}\label{eq:2compete}
        u_2(t_o,  x^*, T,  x^*) = E_1 -  \frac{(T-t_o)^2}{\beta_2} - c_o\delta(x^*) - \frac{r}{\delta(x^*)}
    \end{equation}
By setting  $u_2(t_o, x^*, t_o, x^*) = u_2(t_o,  x^*, T,  x^*)$, we find $T = t_o - \sqrt{(E_1-E_2)\beta_2}$, which is exactly the same as that in \cite{wang2023cooperation}. Therefore, agent 1 must arrive at $t_1 \leq T$ to discourage agent 2 from competing for the better resource. An earlier arrival would increase the traveling cost, making $T$ the optimal arrival time for agent 1. Thus, when agents are competing, the only SPE is $(T, x^*, t_o, x^*)$.

Last, we identify the conditions under which agent 2 prefers to compete with agent 1 by arriving earlier, rather than cooperating, for each of the route selection preference conditions.
\begin{enumerate}
    \item If $c_o \delta_1^2 < \frac{r}{2\lambda}$, we need to compare the following for agent 2:
    \begin{equation}
        u_2(t_o, 2, t_o, 2) = E_2 - c_o\lambda \delta_1 - \frac{r}{2 \lambda \delta_1}
    \end{equation}
    \begin{equation}
        u_2(t_o, 2, t_o-\epsilon, 2) = E_1 - \frac{\epsilon^2}{\beta_2} - c_o\lambda \delta_1 - \frac{r}{ \lambda \delta_1}
    \end{equation}
    \begin{itemize}
        \item When $E_1-E_2 \leq \frac{r}{2\lambda \delta_1}$, then agent 2 will prefer to cooperate as $ u_2(t_o, 2, t_o, 2) \geq u_2(t_o, 2, t_o-\epsilon, 2)$. Thus $(t_o, 2, t_o, 2)$ is the unique SPE.
        \item Otherwise, $(T, 2, t_o, 2)$ is the unique SPE.
    \end{itemize} 
     \item If $c_o \delta_1^2 = \frac{r}{2\lambda}$, the situation resembles the previous case, but with a key difference: in this scenario, agents are indifferent between the two routes in the absence of competition. As a result:
     \begin{itemize}
        \item When $E_1-E_2 \leq \frac{r}{2\lambda \delta_1}$, 
        both $(t_o, 1, t_o, 1)$ and  $(t_o, 2, t_o, 2)$ are the SPEs.
        \item Otherwise, $(T, 2, t_o, 2)$ is the unique SPE.
    \end{itemize}
     \item If $c_o \delta_1^2 \in  (\frac{r}{2\lambda}, \frac{r}{\lambda})$, we need to compare the following for agent 2:
    \begin{equation}
        u_2(t_o, 1, t_o, 1) = E_2 - c_o \delta_1 - \frac{r}{2 \delta_1}
    \end{equation}
    \begin{equation}
        u_2(t_o, 2, t_o-\epsilon, 2) = E_1 - \frac{\epsilon^2}{\beta_2} - c_o\lambda \delta_1 - \frac{r}{ \lambda \delta_1}
    \end{equation}
    \begin{itemize}
        \item When $E_1-E_2 \leq (\lambda-1)c_o\delta_1 - \frac{(\lambda - 2)r}{2 \lambda \delta_1}$, then agent 2 will prefer to cooperate as $ u_2(t_o, 1, t_o, 1) \geq u_2(t_o, 2, t_o-\epsilon, 2)$. Thus $(t_o, 1, t_o, 1)$ is the unique SPE.
        \item Otherwise, $(T, 2, t_o, 2)$ is the unique SPE.
    \end{itemize} 
     \item If $c_o \delta_1^2 = \frac{r}{\lambda}$, the situation resembles the previous case, but with a key difference: in this scenario, agents are indifferent between the two routes in the presence of competition. As a result:
     \begin{itemize}
        \item When $E_1-E_2 \leq (\lambda-1)c_o\delta_1 - \frac{(\lambda - 2)r}{2 \lambda \delta_1}$, $(t_o, 1, t_o, 1)$ is the unique SPE.
        \item Otherwise, both $(T, 2, t_o, 2)$ and $(T, 1, t_o, 1)$ are the SPEs.
    \end{itemize}
    \item If $c_o \delta_1^2 > \frac{r}{\lambda})$, we need to compare the following for agent 2:
    \begin{equation}
        u_2(t_o, 1, t_o, 1) = E_2 - c_o \delta_1 - \frac{r}{2 \delta_1}
    \end{equation}
    \begin{equation}
        u_2(t_o, 1, t_o-\epsilon, 1) = E_1 - \frac{\epsilon^2}{\beta_2} - c_o \delta_1 - \frac{r}{  \delta_1}
    \end{equation}
    \begin{itemize}
        \item When $E_1-E_2 \leq  \frac{r}{2 \delta_1}$, then agent 2 will prefer to cooperate as $ u_2(t_o, 1, t_o, 1) \geq u_2(t_o, 1, t_o-\epsilon, 1)$. Thus $(t_o, 1, t_o, 1)$ is the unique SPE.
        \item Otherwise, $(T, 1, t_o, 1)$ is the unique SPE.
    \end{itemize} 
\end{enumerate}

\end{proof}

\end{document}